\documentclass[conference]{IEEEtran}
\usepackage{ifpdf} 
\usepackage{array} 

\usepackage{amsmath} 
\usepackage{amsthm} 

\theoremstyle{remark}
\newtheorem{theorem}{Theorem}
\newtheorem{corollary}{Corollary}
\newtheorem{lemma}{Lemma}

\newtheorem{proposition}{Proposition}
\newtheorem{remark}{Remark}

\theoremstyle{definition}

\usepackage{amssymb}    %
\usepackage{amsfonts}
\usepackage{cite}
\usepackage{multirow}

\usepackage{makecell}
\usepackage{stfloats}    
\usepackage{algorithm}   
\usepackage{algorithmic} 
\usepackage{graphicx}
\usepackage[nooneline]{caption}
\usepackage{epsfig}
\usepackage{subfigure} 
\usepackage{cite}
\usepackage{tablefootnote}
\usepackage{diagbox}   

\usepackage{setspace} 
\usepackage{threeparttable}

\usepackage{url}

\usepackage{footnote}
\usepackage[textsize=tiny]{todonotes}
\usepackage{marginnote}
\usepackage{pifont}
\usepackage[perpage]{footmisc}  

\usepackage{booktabs}

\allowdisplaybreaks

\usepackage[draft]{changes}   

\usepackage{textcomp}
\usepackage{xcolor}
\def\BibTeX{{\rm B\kern-.05em{\sc i\kern-.025em b}\kern-.08em
T\kern-.1667em\lower.7ex\hbox{E}\kern-.125emX}}
\begin{document}
\title{Estimating Age of Information Using Finite Order Moments}
\author{\IEEEauthorblockN{Licheng Chen\textsuperscript{1}, {Yunquan Dong}\textsuperscript{1,2},~\IEEEmembership{Member,~IEEE}}
\thanks{This work was supported by the National Natural Science Foundation of China (NSFC) under Grant 62071237 and the open research fund of National Mobile Communications Research Laboratory, Southeast University, under grant No. 2020D09.}
	
	\IEEEauthorblockA{\textsuperscript{1}School of Electronic \& Information Engineering, Nanjing University of Information Science \& Technology, Nanjing, China}
	\IEEEauthorblockA{\textsuperscript{2}National Mobile Communications Research Laboratory, Southeast University, Nanjing, China}
	\IEEEauthorblockA{\{charliechen, yunquandong\}@nuist.edu.cn}
}
\maketitle
\begin{abstract}
	Age of information (AoI) has been proposed as a more suitable metric for characterizing the freshness of information than traditional metrics like delay and throughput. However,  the calculation of AoI requires complex analysis and strict end-to-end synchronization. Most existential AoI-related works have assumed that the statistical characterizations of the arrival process and the service process are known.  In fact, due to the randomness of the sources and the channel noises, these processes are often unavailable in reality.
	To this end, we propose a method to estimate the average AoI on a point-to-point wireless Rayleigh channel, which uses the available finite order statistical moments of the arrival process.
	Based on this method, we explicitly present the upper and lower bounds on the average AoI of the system. Our results show that
	1) with the increase of the traffic intensity, the absolute error of the estimated average AoI bounds is first increasing and then decreasing, while the average AoI is monotonically increasing;
	2) the average AoI can be effectively approximated by using the first two order moment estimation bounds, especially when traffic intensity is small or approaches unity;
	3) tighter bounds can be obtained by using more moments.
\end{abstract}

\section{Introduction} \label{section1}
\IEEEPARstart{L}{ow-latency} wireless communications becomes more and more important in the Internet of Things nowadays.
For example, sensors provide timely data on the physical condition of patients in smart healthcare \cite{patient}, vehicles share their own status information (eg. acceleration, location, and so on) in real-time smart driving \cite{AoIputup,vechile}, and so on. These applications have strict timeliness requirements since processors need them to make decisions and control the system \cite{decision,AccessScheme}. However, controlling strategies that focus on traditional metrics, such as improving the throughput to ensure full utilization of the system, or limiting the update rate to reduce the packet delay, are not different from timely scheduling \cite{AoIzerowait}. This is because the former leads to the updates backlog in the queuing system, while the latter causes a lack of fresh data at the receiver.

To this end, the Age of information (AoI) was proposed in \cite{AoIputup}, which is defined as the difference between the current time and the generation time of the latest successfully transmitted packet. As an end-to-end metric, AoI comprehensively characterizes the information staleness due to queue congestion and idle time. Moreover, the performances of lossless first-come-first-served (FCFS) systems and lossy last-come-first-served (LCFS) systems in \cite{AoILGFS} are comparable through AoI since the metric is independent of packet loss \cite{lossy},\cite{AoIpacketsdropping}.

Based on queueing theory, the average AoI of M/M/1, M/D/1, and D/M/1 \cite{AoIzerowait} has been investigated under different service disciplines, e.g. the first-generate-first-served (FCFS) policy \cite{AoIzerowait} and the last-generate-first-served (LGFS) policy \cite{AoILGFS}. The AoI distribution has also been derived for the bufferless system \cite{AoIdistribution}. The authors gave a general extensive formula of the stationary distribution of AoI in \cite{General}, which is suitable for point-to-point systems. However, these serving disciplines cannot be generally optimal for various applications without packet scheduling schemes. The reasonability of the zero-wait policy was investigated in \cite{AoIoptimizing}. Specifically, a packet should be served immediately when the channel becomes free to maximize the throughput and minimize the delay. The average AoI of replacing the first packet in the buffer with newly arriving packets was studied in \cite{Aoipacketsreplacing}.

For the above calculation of average AoI, however, we need to know the arrival and service processes, while these variables are difficult to obtain or are obtained with unknown errors \cite{practivesummary}.
 Thus, many researchers tried to use the generated and received statistical timestamps to obtain the average AoI. In \cite{AoIpractice2018}, the variation of AoI has been investigated for a realistic communication over TCP/IP links served by WiFi, LTE, 3G, 2G, and Ethernet \cite{qwang,swan}. Moreover, the first reported investigation of AoI on real IoT testbeds was present in \cite{AoIpractice2020}. In \cite{AoIpractice2019}, the authors showed the average \text{AoI}-related values of the data streaming are affected by clock synchronization errors between the transmitter and receiver. They also provided an instantaneous \text{AoI} measurement method without clock synchronization. Specifically, the source node sends the packet over the Internet through UDP connections to a receiver, which then echoes back \cite{xchen,zyao}. However, the method requires extra transmissions and contains more uncertainties.

	In this paper, therefore, we propose a method to estimate the average AoI by using a finite number of moments of the arrival and service processes, which provides a pair of tight upper and lower bounds. We obtain the moments of inter-arrival time and service time at two ends without synchronization. Moreover, the estimation of average AoI with this method requires no complex queueing analysis but some simple iterations. Even when only the first and second moments are employed, the estimated bounds provided by our method perform well. Running more iterations using the higher-order moments can further improve the accuracy of the average AoI estimation even more.

The contributions of this paper are given as follows:
\begin{itemize}
\item We provide a discrete transmission model for communications over a rayleigh channel. We estimate its average AoI by using the moments of the arrival process.
\item We provide an estimation method of the average AoI by using finite order moments of the arrival process. The method provides a pair of tight upper and lower bounds for each traffic intensity $\rho$.
\item We show that the performance of our method depends on the magnitude of each order moment and traffic intensity. We also present the distributions of the inter-arrival time achieving the largest and the smallest average AoI for the system.
\end{itemize}
We organize the rest of the paper as follows. In Section \ref{section 2}, we present the system model and the transmission models. Section \ref{section 4} discusses the transmission model and presents its average AoI and corresponding estimation method. In Section \ref{section 5}, we test the performance of our method in estimating the average AoI by simulations. Finally, we conclude our work in Section \ref{section 6}.	
\section{SYSTEM MODEL} \label{section 2}
\subsection{Channel and Transmission Model}\label{Channel and Transmission Model}
\begin{figure}[!t]
\centering
\includegraphics[width=3.45in]{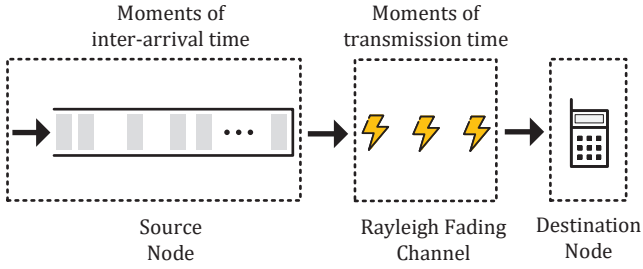}
\caption{Point-to-point communication system}
\label{fig1}
\end{figure}
We consider a wireless communication system consisting of a source node and a destination node, where the source delivers its packets over a block fading channel with additive white Gaussian noise (AWGN), as shown in Fig. \ref{fig1}. Time is discrete and the length of a block is $T_{\text{B}}$. The source node samples the surrounding environment with random intervals. Each sample is encoded into a packet of $L$ bits and stored in an infinite-long buffer. The packets are then transmitted over the fading channel by using the FCFS discipline. We consider the packet transmission over a fading Rayleigh channel with power gain distribution $f_{\gamma}(x)$, which is assumed to be known.
Based on a sequence of packet transmissions, we assume that the moments of the transmission time of the packets can be estimated up to some finitely large order.

In each block, the transmission of a packet will be successful if the instantaneous signal-to-noise ratio (SNR) is larger than a certain threshold $V_T$. Otherwise, the packet shall be retransmitted in the next block (c.f. Section \ref{section 4}). It can be readily shown that the transmission time of a packet follows the geometric distribution. We assume that the distribution of the inter-arrival time is unknown, and estimate the average AoI with its moments.
\subsection{Queueing system}
We model the transmission process over the channel as a single-server FCFS queue.
As shown in Fig. \ref{fig2}, we denote the arrival (generation) epoch and the departure epoch of the $k^{\text{th}}$ packet, respectively, as $n_k$ and $n_k'$.
We denote the period between two neighboring arrival epochs as inter-arrival time $X_k=n_k-n_{k-1}$ and denote the period between two neighboring departures as inter-departure time $Y_k=n'_k-n'_{k-1}$.

The number of blocks required to complete a packet transmission is referred to as the service time. Note that the service time $S_k$ of the packets is determined by the channel and is independent of inter-arrival times.
We also denote the system time as ${T}_{k}=n_{k}'-n_{k}$. If the $k^{\text{th}}$ packet arrives at the buffer when the source is busy transmitting the $(k-1)^{\text{th}}$ packet, it must wait in the buffer until the transmission of the $(k-1)^{\text{th}}$ packet is completed. Thus, the waiting time of the $k^{\text{th}}$ packet is $W_k=\max\{n_{k-1}'-{n_k},0\}=\max\{T_{k-1}-X_k,0\}$ and we have $T_k=W_k+S_k$. On the contrary, the buffer remains empty for a certain duration if the arrival epoch of the $k^{\text{th}}$ packet is later than the departure epoch of the $k-1^{\text{th}}$ packet. Thus, the length of this period can be expressed as ${I_k}=\max\{n_{k}-n_{k-1}^{'},0\}=\max\{{X}_{k}-{T}_{k-1},0\}$ and we have $Y_k=I_k+S_k$.

We refer to the ratio of the expectation between service time and the inter-arrival time as traffic intensity, i.e., $\rho={E(S)}/{E(X)}$. The system is stable if $\rho<1$.
\subsection{Age of Information}
\begin{figure}[htbp!]
\centering
\includegraphics[width=3.5in]{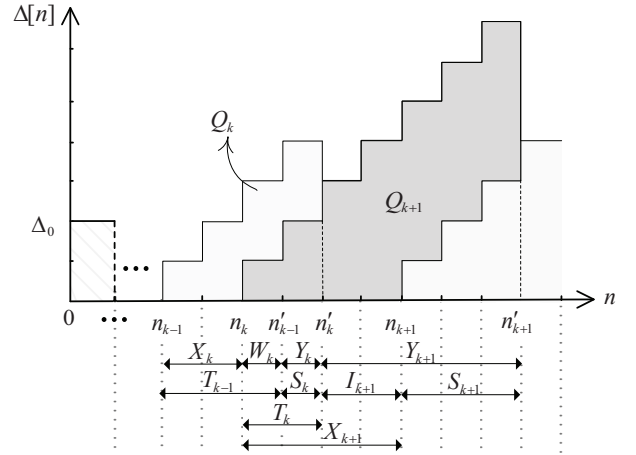}
\caption{Sample path of discrete average \text{AoI}  $\Delta(n)$}
\label{fig2}
\end{figure}
In block $n$, we denote the arrival epoch of the latest successfully received packet as $ U(n) $ and the \text{AoI} $\Delta(n)$ is defined as
\begin{equation*}
\Delta(n)=n-U(n).
\end{equation*}
	From Fig. \ref{fig2}, it can be seen that $\Delta(n)$ first increases (up to  $\Delta({n'_{k}}^-)=X_k+T_k$) and then drops to $\Delta(n'_{k})=T_{k}$ upon receiving a new packet at the destination.
During a period of $N$ blocks, we assume that there are $K$ successfully received packets and denote the average AoI of this period as ${\bar\Delta}=\lim_{N\to \infty} \sum\nolimits_{n=0}^{N}{\Delta(n)}/N$.
By dividing the area of Fig. \ref{fig2} into a sequence of disjoint polygon areas which combine to form the area under the AoI curve, we have
\begin{equation}\label{AoI}
\bar\Delta=\frac{E(XT)+E(X^2)/2}{E(X)},
\end{equation}
	in which $Q_k=\tfrac{1}{2}(T_k+X_k)^2-\tfrac{1}{2}T_k^2$ [\ref{survey}, (3)].
\begin{remark}
To calculate the term $E(XT)$ in the expression of the average AoI shown in \eqref{AoI}, we need the probability density/mass distributions of the inter-arrival time $X_k$.
In practical implementations, it is difficult to obtain the distribution of inter-arrival time. On the other hand, it is much easier to obtain its moments.
In this paper, therefore, we shall use the moments of inter-arrival time $X_k$ to estimate the average AoI of the system.
In particular, it is assumed that the needed moments have been estimated and are available for our analysis.
\end{remark}
\section{Bounding Average AoI of the Transmission Model}\label{section 4}
Under this model (cf. Section \ref{Channel and Transmission Model}), the destination node can decode the packet from the received signal if the instantaneous SNR $\text{snr}_\text{n}$ is greater than a certain threshold $V_{\text{T}}$.
In case the packet cannot be decoded, and the source node needs to retransmit the packet again in the next block.
We represent the probability of a successful packet transmission as $\mu$, i.e.,
\begin{equation*}
\mu =\Pr\{\text{snr}_\text{n} >V_{\text{T}}\}.
\end{equation*}
Note that $\mu$ can be estimated based on the results in channel estimations.
We denote the number of blocks required for the successful transmission of the $k^{\text{th}}$ packet as service time $S_k$.
	It is clear that $S_k$ is a geometric random variable with the parameter $\mu$ regardless of the fading model.
Specifically, the distribution of $S_k$ can be expressed by
\begin{align}
\Pr\{S_k=n\} &={{\bar{\mu }}^{n-1}}\mu, &n\ge 1,\label{pdf_s2}
\end{align}
in which $\bar \mu=1-\mu$.

We make the following assumptions on the packet process.
\begin{enumerate}
\item[\textit{D1.}]  \textit{Distribution unavailable:} The distribution of the inter-arrival time $X_k$ is unknown \cite{channeuncertain}.
\item[\textit{D2.}]  \textit{Moments available:} The first order moment $E(X_k)$ and the second order moment $E(X_k^2)$ of the inter-arrival time exist and are known.
\item[\textit{D3.}] \textit{Early arrival system:} The packets arrive at the end of epochs $(n^-,n)$. The transmission of each packet starts and completes at the beginning of blocks $(n',{n'}^+)$ [\ref{alpha3} pp. 193].
\end{enumerate}

\subsection{Characterizing Average AoI Using Moments and PGF}
We denote the number of delivered packets during the inter-arrival time $X_k$ as $B_k$, and denote the queue length at the arrival of the $k^{\text{th}}$ packet as $L_k^-$. It is clear that $L^-_{k}$ can be expressed as
\begin{equation*}
L^-_{k+1}=\left\{
\begin{aligned}
&L^-_{k}+1-B_{k+1}   & L^-_{k}+1-B_{k+1}>0, \\
&0      & L^-_{k}+1-B_{k+1}\le0.
\end{aligned}
\right.
\end{equation*}
The process $\{{L^-_k,k\ge0}\}$ is an embedded Markov chain.
The stationary distribution of queueing length $L^-$ is explicitly present in the following lemma.
\begin{lemma}
	Given $\rho<1$ and the system is stable, the stationary distribution of queue length $L^{-}=\lim_{k\to\infty} L^{-}_k$  exists and can be expressed by
	\begin{align*}
\Pr\{L^{-}=n\}&=(1-\alpha)\alpha^n\, &n\ge0,
	\end{align*}
where $\alpha$ is the unique real root of equation $z=G_X(\bar\mu+{\mu}z)$ for $z\in(0,1)$. Furthermore, the distribution of the system time $T_k$ of a packet is
\begin{align}\label{PGF_T2}
\Pr\{T=n\}&=(1-\lambda ){{\lambda }^{n-1}},&n\ge 1.
\end{align}
\end{lemma}
	\begin{proof}
A detailed proof can be found in [\ref{queueing1}, Section 4.6.2].
	\end{proof}

Based on \eqref{PGF_T2}, the average AoI of the system can be expressed as shown in the following theorem.
\begin{theorem}\label{gi/geom/1 AoI}
The average AoI of the system can be expressed as
\begin{equation}\label{AoI2}
\bar\Delta_{\text{d}}=E(S)+\frac{E(X^2)}{2E(X)}+\frac{\lambda G_X'(\lambda) }{E(X)(1-\lambda)}.
\end{equation}
\end{theorem}

In \eqref{AoI2}, $E(S)$ can be calculated or estimated.
From assumption D1, the first two order moments $E(X)$ and $E(X^2)$ are also available through some estimations.
The remaining problem is whether we can approximate $G_X'(\lambda)$ by using the first two order moments $E(X)$ and $E(X^2)$.
Moreover, we also need to solve $\alpha$ and $\lambda$ according to $G_X(\lambda)=\alpha$ and $\lambda=\bar \mu+\mu\alpha$. Assuming that $\rho<1$ and the queue being stable, $\lambda$ can be obtained by
$\lambda=\frac{\bar{u}}{1-\Pr\{Y=1\}}$,
where $\Pr\{Y=1\}$ is the probability of the departure interval taking on the value 1.

\subsection{Estimating Average AoI Using First Two Order Moments}
The PGF of a discrete random variable $X$ exists only if $X$ has finite moments $E(S^n)$ for all $n=1, 2,\cdots$ (cf. [\ref{moment2}, pp. 1]).
The distribution of $X$ is determined exclusively by the moment sequence [\ref{uniquelydetermine} Theorem 4.17.1]. If some of the moments of inter-arrival time $S$ are not finite. We can truncate the support of variable $X$ as in [\ref{truncated}, Section 3.3]. We assume that the inter-arrival time has limited moments without losing generality.

We present the PGF of the inter-arrival time with its moments, as shown in the following theorem.
\begin{theorem} \label{Th5}
The first order PGF $G'_X(z)$ of the inter-arrival time can be expressed by its moments as
\begin{align}\label{G'_X(z)}
&G'_X(z)=\frac{1}{z}\sum_{n=1}^\infty {\frac{\ln (z)^{n-1}E(X^n)}{(n-1)!}}, &z\in(0,1].
\end{align}
\end{theorem}
\begin{proof}
By the definition of PGF, we have
\begin{align*}
G'_X(z)=\frac{1}{z}\sum_{m\ge 1}{m\Pr\{X=m\}z^{m}}.
\end{align*}
By expressing  $z^m=\exp(m\ln (z))$ in its Taylor series, we then have
\begin{align*}
G'_X(z)&=\frac{1}{z}{\sum_{m\ge 1}{m\Pr\{X=m\}}\sum_{n\ge 0}{\frac{\ln (z)^{n}m^{n}}{n!}}},\\
&=\frac{1}{z}\sum_{n\ge0}{\frac{\ln (z)^n}{n!}}\sum_{m\ge1}{\Pr\{X=m\}m^{n+1}},\\
&=\frac{1}{z}\sum_{n\ge 1} {\frac{\ln (z)^{n-1}E(X^n)}{(n-1)!}}.
\end{align*}
Therefore, Theorem \ref{Th5} is proved.
\end{proof}
We denote the partial sum of the PGF given in \eqref{G'_X(z)} as
\begin{equation}\label{G_X'(z)_p}
\widehat G_X^{'K}(z)=\frac{1}{z}\sum_{n=1}^K {\frac{\ln (z)^{n-1}E(X^n)}{(n-1)!}},
\end{equation}
in which $K=1,2,\cdots$ and $z\in(0,1]$. The item $G'_X(\lambda)$ can be approximated by the partial sum in \eqref{G_X'(z)_p} with an acceptable error.
\begin{lemma}
For the random variable $X$ whose PGF exist and probability $\Pr\{X=0\}=0$, the following bounds hold
\begin{equation} \label{initial_bound2}
\begin{aligned}
	&\frac{G_X(z)}{z} \leq G'_X(z) \leq \frac{1-G_X(z)}{1-z},&z\in(0,1].	
\end{aligned}
\end{equation}
\end{lemma}
\begin{proof}
In Fig. \ref{bounds of G}, we take the case $z=\lambda$ for example. Since $G'_X(z)$ is a monotonically increasing concave function, we can apply Jensen's inequality in \cite{Jenseninequality} to derive the constraints
$l'_{L_0}(\lambda)\le G'_X(\lambda) \le l'_{U_0}(\lambda)$,
where $l'_{L_0}(\lambda)={G_X(\lambda)}/{\lambda}=\alpha/\lambda$ and $l'_{U_0}(\lambda)=\frac{1-G_X(\lambda)}{1-\lambda}=E(S)$.

\begin{figure}[htbp!]
\centering
\includegraphics[width=3.5in]{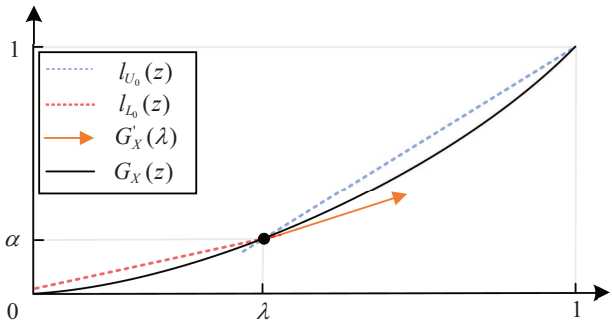}	
\caption{The initial bounds of $G'_X(\lambda)$.}
\label{bounds of G}
\end{figure}
\end{proof}	
Thus we can use $\widehat G_{X}^{'1}(\lambda)$ and $\widehat G_{X}^{'2}(\lambda)$ to estimate the average AoI of the system instead of using $G_{X}'(\lambda)$.
\begin{theorem}
Based on the estimations of the first two order moments of $X_k$, the bounds of the average AoI of the system can be expressed as
\begin{subequations}\label{rt:aoi_d_bound}
\begin{align}f
	\bar\Delta_{\text{d}}&<E(S)+\frac{E(X^2)}{2E(X)}+\frac{\lambda \min\{E[S],G_{X}^{'1}(\lambda)\} }{E(X)(1-\lambda)},\\
	\bar\Delta_{\text{d}}&>E(S)+\frac{E(X^2)}{2E(X)}+\frac{\lambda \max\{\alpha/\lambda,G_{X}^{'2}(\lambda)\} }{E(X)(1-\lambda)}.
\end{align}
\end{subequations}
\end{theorem}
\begin{proof}
The theorem is proved by substituting the item $G_{X}'(\lambda)$ in \eqref{AoI2} by $G_{X}^{'1}(\lambda)$ and $G_{X}^{'2}(\lambda)$.
\end{proof}
We also demonstrate how the partial sum limits the PGF $G'_{X}(z)$ in the following proposition.
\begin{proposition}\label{G'X_odd_even}
Each partial sum estimation $\widehat G_X^{'K}(z)$ (cf. \eqref{G_X'(z)_p}) up to an even number of orders is smaller than $G'_X(z)$, and each partial sum estimation up to an odd number of orders is greater than $G'_X(z)$.
That is,
\begin{subequations}\label{G'X_odd_even_compare}
\begin{align}
	\widehat{G}{^{'2i+1}_{X}}(z)& > G'_{X}(z) \\
	\widehat {G}{^{'2i}_{X}} (z)&< G'_{X}(z),
\end{align}
\end{subequations}
for any $z\in (0,1]$ and $i=0,1,2,\cdots$.
\end{proposition}

Tighter bounds can be obtained by minimizing the upper bound and maximizing the lower bound in \eqref{G'X_odd_even_compare} over the order of moments (equivalently over $i$). Additionally, the bounds in \eqref{initial_bound2} apply to distributions that have PGF. Consequently, for $z=\lambda$, we have
\begin{subequations}\label{G_X_AoI_bound}
\begin{align}
G'_X(\lambda)&\leq \underset{i}{\min} \left\{\widehat {G}_{X}^{'2i+1}(\lambda),E(S)\right\} \doteq  {\overset{\frown}{G^{'2i+1}_{X}}}(\lambda),\\
G'_X(\lambda)&\geq\underset{i}{\max}\left\{\widehat G_X^{'2i}(\lambda),\tfrac{\alpha}{\lambda}\right\} \doteq {\underset{\smile}{G^{'2i}_{X}}}(\lambda).
\end{align}
\end{subequations}

By combining \eqref{AoI2} and \eqref{G_X_AoI_bound}, we show a pair of tighter and more inclusive bounds on the average AoI of the system in the following corollary.
\begin{corollary}\label{lem_est_age_1}
Based on the higher $2i+1^{\text{st}}(i>0)$ order moments, the average AoI of the discrete transmission model is upper and lower bounded, respectively, by
\begin{subequations}\label{rt:aoi_dN_bound}
\begin{align}
	\bar\Delta_{\text{d}}&<E(S)+\frac{E(X^2)}{2E(X)}+\frac{\lambda {\overset{\frown}{G^{'2i+1}_{X}}}(\lambda) }{E(X)(1-\lambda)}\doteq{\overset{\frown}{\Delta}{^{2i+1}_d}},\\
	\bar\Delta_{\text{d}}&>E(S)+\frac{E(X^2)}{2E(X)}+\frac{\lambda {\underset{\smile}{G^{'2i}_{X}}}(\lambda) }{E(X)(1-\lambda)}\doteq{\underset{\smile}{\Delta}{^{2i}_d}},
\end{align}
\end{subequations}
where ${\overset{\frown}{G^{'2i+1}_{X}}}(\lambda)$ and ${\underset{\smile}{G^{'2i}_{X}}}(\lambda)$ are given in \eqref{G_X_AoI_bound}.
\end{corollary}

\section{NUMERICAL RESULTS}\label{section 5}
\begin{figure}[!htp]
\centering
\subfigure[{the absolute error of the estimated average AoI bounds versus traffic intensity $\rho$, in which arrival rate $p=0.125$}]
{
\begin{minipage}{0.45\textwidth}\label{fig_GI_GEOM_a}
	\centering
	\includegraphics[width=3.4in]{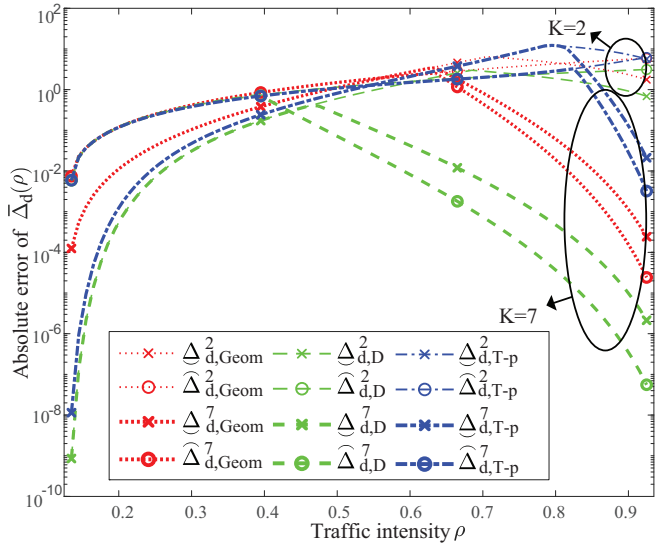}
\end{minipage}}
\subfigure[{Average AoI versus traffic intensity $\rho$, in which arrival rate $p=0.125$}]
{\begin{minipage}{0.45\textwidth}\label{fig_GI_GEOM_b}
	\centering
	\includegraphics[width=3.4in]{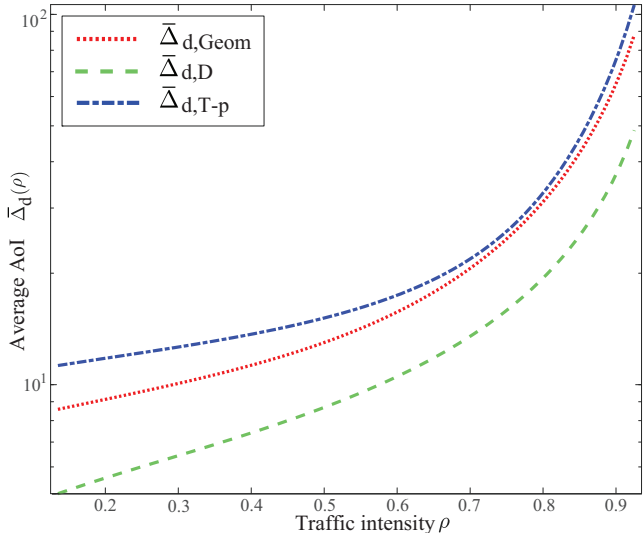}
\end{minipage}}
\caption{{Numerical results of system}}
\label{fig_GI_GEOM}
\end{figure}

In this section, we investigate the accuracy of the average AoI bounds under various sets of system parameters, such as the distribution of inter-arrival time, the maximum order of moments used for the estimation, and the traffic intensity $\rho$, in which
\begin{itemize}
\item the traffic intensity $\rho=p/\mu$ determines the estimated item $G_{\text{X}}'(\lambda) $ (cf. \eqref{AoI2});
\item  the estimation errors for distributions with larger high order moments are also larger for a given common expectation;
\item as shown in \eqref{G'_X(z)}, the estimation error goes to zero when an infinite number of moments are utilized, and it would most probably be reduced when more orders of moments are used.
\end{itemize}
Given a common expectation, we evaluate the estimation errors of the bounds for the degenerate distribution, the two-point distribution, and the geometric distribution.

Fig. \ref{fig_GI_GEOM} shows how the average AoI $\bar \Delta$ and its corresponding estimation absolute error change with traffic intensity $\rho$. We set the arrival rate as $p=1/E(X)=0.125$ and change the service rate $\mu$ so that the traffic intensity takes values between zero and unity.
In Fig. \ref{fig_GI_GEOM_a}, we show the absolute error of the lower and upper bounds for the average AoI by the curves marked by $\times$ and $\circ$, respectively.
First, we observe that the absolute errors are first increasing and then decreasing as $\rho$ increases. When $\rho$ approaches unity, the average AoI is large, and when $\rho$ is small, the absolute error is small.
Second, it is observed from Fig. \ref{fig_GI_GEOM_b} that the average AoIs monotonically increase in $\rho$. When $\rho$ approaches unity, the average AoI is large, and when $\rho$ is small, the absolute error is small. So we observe that the errors of the bounds are relatively small in these two cases.
Third, for the lower and the upper bounds, the absolute estimation error of the two-point distribution is the largest while the absolute estimation error of the degenerate distribution is the smallest.
Fourth, the thin and the thick curves characterize the absolute errors obtained with the first two (i.e., $K=2$) order moments and with up to the seventh (i.e., $K=7$) order moments, respectively.
When more moments are used, the absolute estimation error is substantially decreased as $\rho$ approaches unity.
Moreover, when $\rho$ is small, the boundaries overlap no matter whether we use the first two order moments or up to the seventh order moments.
Among the three tested distributions, the separations of curves of the two-point distribution require the largest $\rho$.
Due to the small $\lambda$ being given in \eqref{PGF_T2} introduced by the more minor $\rho$, the bounds using the first two order moments and bounds using the first seven order moments both perform worse than the initial set which replace them, so the first half of these bounds overlap.
\section{CONCLUSION}\label{section 6}
\noindent In this paper, we investigated the problem of estimating the average AoI in the point-to-point wireless networks. Specifically, we consider a typical transmission model and obtain its average AoI. The unknown term in the expression is PGF of the arrival process. We explicitly present upper and lower bounds on the average AoI through using finite order moments of the arrival processes. By using only the first two order moments of inter-arrival time, we can derive a pair of tight bounds. Using more higher-order moments can improve the performance of the calculation. Through numerical simulations, we investigated the effect of the traffic intensity, the distribution of the inter-arrival time and magnitude of each order moment on the performance of our estimation method.


\begin{thebibliography}{00}
\bibitem{patient} Ko J G, Lu C, Srivastava M B, et al., ``Wireless sensor networks for healthcare,'' Proceedings of the IEEE, 2010, 98(11): 1947-1960.
\bibitem{AoIputup}\label{AoIputup} S. Kaul, M. Gruteser, V. Rai and J. Kenney, ``Minimizing age of information in vehicular networks,'' 2011 8th Annual IEEE Communications Society Conference on Sensor, Mesh and Ad Hoc Communications and Networks, 2011, pp. 350-358.
\bibitem{vechile} Qiong Wu, Yu Zhao and Qiang Fan, “Time-Dependent Performance Modeling for Platooning Communications at Intersection”, published online, IEEE Internet of Things Journal, Mar. 2022.
\bibitem{decision} Bao Z, Dong Y, Chen Z, et al. ``Age-optimal service and decision processes in Internet of Things,'' IEEE Internet of Things Journal, 2020, 8(4): 2826-2841.
\bibitem{AccessScheme} Qiong Wu, Ziyang Wan, Qiang Fan, Pingyi Fan and Jiangzhou Wang, “Velocity-adaptive Access Scheme for MEC-assisted Platooning Networks: Access Fairness Via Data Freshness,” IEEE Internet of Things Journal, Vol. 9, No. 6, Mar. 2022, pp. 4229-4244.
\bibitem{AoIzerowait}\label{AoIzerowait} S. K. Kaul, R. D. Y ates, and M. Gruteser, ``Real-time status: How often should one update?,'' in Proc. IEEE INFOCOM, Orlando, FL, USA, Mar. 2012, pp. 2731–2735.
\bibitem{AoILGFS} A. M. Bedewy, Y . Sun, and N. B. Shroff, ``Optimizing data freshness,throughput, and delay in multi-server information-update systems,'' in Proc. IEEE Int. Symp. Inf. Theory (ISIT), Barcelona, Spain, Jul. 2016,pp. 2569–2573.
\bibitem{lossy} Yates R D. ``The age of information in networks: Moments, distributions, and sampling,'' IEEE Transactions on Information Theory, 2020, 66(9): 5712-5728.
\bibitem{AoIpacketsdropping} M. Costa, M. Codreanu, and A. Ephremides, ``On the age of information in status update systems with packet management,'' IEEE Transactions on Information Theory, vol. 62, no. 4, pp. 1897–1910, Apr. 2016.
\bibitem{AoIdistribution} Kesidis G, Konstantopoulos T, Zazanis M A. ``The distribution of age-of-information performance measures for message processing systems,'' Queueing Systems, 2020, 95(3): 203-250.
\bibitem{General} Y. Inoue, H. Masuyama, T. Takine and T. Tanaka, ``A General Formula for the Stationary Distribution of the Age of Information and Its Application to Single-Server Queues,'' IEEE Transactions on Information Theory, vol. 65, no. 12, pp. 8305-8324, Dec. 2019.
\bibitem{AoIoptimizing} Y . Sun, E. Uysal-Biyikoglu, R. D. Yates, C. E. Koksal, and N. B. Shroff,``Update or wait: How to keep your data fresh,'' IEEE Transactions on Information Theory, vol. 63, no. 11, pp. 7492–7508, Nov. 2017.
\bibitem{Aoipacketsreplacing} C. Kam, S. Kompella, G. D. Nguyen, J. E. Wieselthier, and A. Ephremides,``Controlling the age of information: Buffer size, deadline, and packet replacement,'' in Proc. IEEE Mil. Commun. Conf., Baltimore,MD, USA, Nov. 2016, pp. 301–306
\bibitem{practivesummary}\label{practivesummary} E. Uysal, O. Kaya, S. Baghaee, and H. B. Beytur, ``Age of information in practice,'' CoRR, vol. abs/2106.02491, 2021. [Online]. Available: https://arxiv.org/abs/2106.02491
\bibitem{AoIpractice2018} C. Sönmez, S. Baghaee, A. Ergişi and E. Uysal-Biyikoglu, ``Age-of-Information in Practice: Status Age Measured Over TCP/IP Connections Through WiFi, Ethernet and LTE,'' 2018 IEEE International Black Sea Conference on Communications and Networking, 2018, pp. 1-5.
\bibitem{AoIpractice2020} Beytur H B, Baghaee S, Uysal E. ``Towards AoI-aware smart IoT systems,'' 2020 International Conference on Computing, Networking and Communications (ICNC). IEEE, 2020: 353-357.
\bibitem{AoIpractice2019} Beytur H B, Baghaee S, Uysal E. ``Measuring age of information on real-life connections,'' 2019 27th Signal Processing and Communications Applications Conference (SIU). IEEE, 2019: 1-4.
\bibitem{survey}\label{survey} Yates R D, Sun Y, Brown D R, et al. ``Age of information: An introduction and survey,'' IEEE Journal on Selected Areas in Communications, 2021, 39(5): 1183-1210.
\bibitem{alpha3}\label{alpha3} Hunter J J. ``Mathematical techniques of applied probability: Volume 2, Discrete time models: techniques and applications,'' Academic Press, 1985.
\bibitem{queueing1}\label{queueing1} Alfa A S. ``Queueing theory for telecommunications: discrete time modelling of a single node system,'' Springer Science \& Business Media, 2010.
\bibitem{moment2}\label{moment2} Lin G D. ``Recent developments on the moment problem,'' Journal of Statistical Distributions and Applications, 2017, 4(1): 1-17.
\bibitem{uniquelydetermine}\label{uniquelydetermine} Simon B. ``A comprehensive course in analysis,'' Providence, Rhode Island: American Mathematical Society, 2015.
\bibitem{truncated}\label{truncated} Wanyang D A I. ``On the conflict of truncated random variable vs. heavy-tail and long range dependence in computer and network simulation,'' Journal of Computational Information Systems, 2011, 7(5): 1488-1499.


\bibitem{qwang}
Q. Wang, D. Wu, and P. Fan, ``Delay-constrained optimal link scheduling in wireless sensor networks,"  \textit{IEEE Transactions on Vehicular Technology},  vol. 59, no. 9, pp. 4564--4577, Sep. 2010.

\bibitem{swan}
S. Wan, J. Lu, P. Fan, and K .B. Letaief, ``To smart city: Public safety network design for emergency," \textit{IEEE access}, vol. 6, pp. 1451--1460, Mar. 2017.

\bibitem{xchen}
X. Chen, J. Lu, P. Fan, and K. B. Letaief, ``Massive MIMO beamforming with transmit diversity for high mobility wireless communications",  vol. 6, pp. 23032--23045, May. 2017.

\bibitem{zyao}
Z. Yao, J. Jiang, P. Fan, Z. Cao, and Vok Li, ``A neighbor-table-based multipath routing in ad hoc," \textit{in Proc. the 57th IEEE Semiannual Vehicular Technology Conference}, 2003.

\bibitem{channeuncertain} Hongbiao Zhu, Qiong Wu, Xiao-Jun Wu, Qiang Fan, Pingyi Fan and Jiangzhou Wang, “Decentralized Power Allocation for MIMO-NOMA Vehicular Edge Computing Based on Deep Reinforcement Learning,” IEEE Internet of Things Journal, Vol. 9, No. 4, Jul 2022, pp. 12770-12782.
\bibitem{Jenseninequality} Jensen J L W V. ``Sur les fonctions convexes et les inégalités entre les valeurs moyennes,'' Acta mathematica, 1906, 30(1): 175-193.
\end{thebibliography}
\end{document}